\documentclass[runningheads,a4paper]{llncs}
\usepackage[utf8]{inputenc}
\usepackage{epsfig}

\usepackage{graphicx}
\usepackage[usenames,dvipsnames]{color}

\newcommand{\nat}{\ensuremath {\mathbb N} }

\newcommand{\Prob}{\mathbb{P}}
\newcommand{\eps}{\varepsilon}
\newcommand{\E}{\mathbb E}

\newcommand{\aas}{{\sl{a.a.s.}}}

\newcommand{\bin}{\mathrm{Bin}}

\usepackage{graphicx}
\usepackage[colorlinks=true, allcolors=blue]{hyperref}

\usepackage{paralist}
\usepackage{amsmath}
\usepackage{amssymb}
\usepackage{algorithmic}
\usepackage{algorithm}
\usepackage{tikz}
\usepackage[]{todonotes}   

\newtheorem{cor}{Corollary}

\author{
Lenar Iskhakov\inst{1}
\and
Bogumi\l{} Kami\'nski\inst{2}
\and
Maksim Mironov\inst{1}
\and 
Liudmila Ostroumova Prokhorenkova\inst{1,3}
\and
Pawe\l{}~Pra\l{}at\inst{4}
}
\institute{Advanced Combinatorics and Network Applications Lab, Moscow Institute of Physics and Technology, Moscow, Russia 
\and
Warsaw School of Economics, Warsaw, Poland
\and
Machine intelligence and research department, Yandex, Moscow, Russia
\and
Department of Mathematics, Ryerson University, Toronto, Canada}

\authorrunning{Lenar Iskhakov et al.}

\title{Clustering Properties of Spatial Preferential Attachment Model}

\begin{document}

\maketitle

\begin{abstract}
In this paper, we study the clustering properties of the Spatial Preferential Attachment (SPA) model introduced by Aiello et al. in 2009. This model naturally combines geometry and preferential attachment using the notion of spheres of influence. It was previously shown in several research papers that graphs generated by the SPA model are similar to real-world networks in many aspects. For example, the vertex degree distribution was shown to follow a power law. In the current paper, we study the behaviour of $C(d)$, which is the average local clustering coefficient for the vertices of degree $d$. This characteristic was not previously analyzed in the SPA model. However, it was empirically shown that in real-world networks $C(d)$ usually decreases as $d^{-a}$ for some $a>0$ and it was often observed that $a=1$. We prove that in the SPA model $C(d)$ decreases as $1/d$. Furthermore, we are also able to prove that not only the average but the individual local clustering coefficient of a vertex $v$ of degree $d$ behaves as $1/d$ if $d$ is large enough. The obtained results are illustrated by numerous experiments with simulated graphs.

\end{abstract}

\section{Introduction}
\vspace{-9pt}

The evolution of complex networks attracted a lot of attention in recent years.  Empirical studies of different real-world networks have shown that such networks have some typical properties: small diameter, power-law degree distribution, clustering structure, and others~\cite{costa2007characterization,newman2003structure}. 
Therefore, numerous random graph models have been proposed to reflect and predict such quantitative and topological aspects of growing real-world networks~\cite{boccaletti2006complex,bollobas2003mathematical}.

The most well studied property of complex networks is their vertex degree distribution. For the majority of studied real-world networks, the degree distribution was shown to follow a heavy-tailed distribution~\cite{barabasi1999emergence,faloutsos1999power,newman2005power}.
Another important property of real-world networks is their clustering structure. One way to characterize the presence of clustering structure is to measure the {\it clustering coefficient}, which is, roughly speaking, the probability that two neighbours of a vertex are connected. There are two well-known formal definitions: the global clustering coefficient and the average local clustering coefficient (see Section~\ref{sec:clustering} for details). It is widely believed that for many real-world networks both the average local and the global clustering coefficients tend to non-zero limit as the network becomes large; some numerical values can be found in~\cite{newman2003structure}; however, some contradicting theoretical results are presented in~\cite{ostroumova2016global}.

In this paper, we mostly focus on the behaviour of $C(d)$, which is the average local clustering coefficient for the vertices of degree $d$. It was empirically shown that in real-world networks $C(d)$ usually decreases as $d^{-\psi}$ for some $\psi>0$~\cite{csanyi2004structure,leskovec2008dynamics,serrano2006clustering,vazquez2002large}. In particular, for many studied networks, $C(d)$ scales as $d^{-1}$~\cite{ravasz2003hierarchical}.

We study the clustering properties of the \emph{Spatial Preferential Attachment} (SPA) model introduced in~\cite{spa1}. This model combines geometry and preferential attachment; 
the formal definition is given in Section~\ref{sec:spa_def}. It was previously shown that graphs generated by the SPA model are similar to real-world networks in many aspects.  For example, it was proven in \cite{spa1} that the vertex degree distribution follows a power law. More details on the properties of the SPA model are given in Section~\ref{sec:spa_prop}. However, the clustering coefficient $C(d)$ was not previously analyzed for this model, although some clustering properties were analyzed for the generalized SPA model proposed in \cite{jacob2013spatial}.
It is proved in \cite{jacob2013spatial} and \cite{jacob2015spatial} that the average local clustering coefficient converges in probability to a strictly positive limit. Also, the global clustering coefficient converges to a nonnegative limit, which is nonzero if and only if the power-law degree distribution has a finite variance.

In this paper, we prove that the local clustering coefficient $C(d)$ decreases as $1/d$ in the SPA model. We also obtain some bounds for the individual local clustering coefficients of vertices. The obtained theoretical results are compared with and illustrated by numerous experiments on simulated graphs. Our theoretical results are asymptotic in nature, so we empirically investigate how the model behaves for finite size graphs and see that the asymptotic predictions are still close to empirical observations even for small graph sizes. Additionally, we demonstrate that some of our theoretical assumptions are probably too pessimistic and the SPA model behaves even more predictable than we have proven. We also propose an efficient algorithm for generating graphs according to the SPA model which runs much faster than the straightforward implementation.

Proofs of all theoretical results stated in this paper can be found in the journal version~\cite{iskhakov2017local} that focuses exclusively on asymptotic results of the model. On the other hand, this proceeding version also contains results on simulated graphs and so can be viewed as a complement to the journal version. 

\section{Spatial Preferential Attachment model}\label{sec:spa}

\subsection{Definition}\label{sec:spa_def}

This paper focuses on the \emph{Spatial Preferential Attachment} (SPA) model, which was first introduced by~\cite{spa1}. This model combines preferential attachment with geometry by introducing ``spheres of influence'' whose volume grows with the degree of a vertex. The parameters of the model are the \emph{link probability} $p\in[0,1]$ and two constants $A_1,A_2$ such that $0 < A_1 < \frac{1}{p}$, $A_2>0$. All vertices are placed in the $m$-dimensional unit hypercube $S = [0,1]^m$ equipped with the torus metric derived from any of the $L_k$ norms, i.e.,  
\[
d(x,y)=\min \big\{ ||x-y+u||_k\,:\,u\in \{-1,0,1\}^m \big\} \,\,\,\,\,\,\, \forall  x,y \in S \,.
\]
The SPA model generates a sequence of random directed graphs $\{G_{t}\}$, where $G_{t}=(V_{t},E_{t})$, $V_{t}\subseteq S$. Let $\deg^{-}(v,t)$ be the in-degree of the vertex $v$ in $ G_{t}$, and $\deg^+(v,t)$ its out-degree. Then, the \emph{sphere of influence} $S(v,t)$ of the vertex $v$ at time $t\geq 1$ is the ball centered at $v$ with the following volume:
$$
|S(v,t)|=\min\left\{\frac{A_1{\deg}^{-}(v,t)+A_2}{t},1\right\}.
$$

In order to construct a sequence of graphs we start at $t=0$ with $G_0$ being the null graph. At each time step $t$ we construct $G_{t}$ from $G_{t-1}$ by, first, choosing a new vertex $v_t$ \emph{uniformly at random} from $S$ and adding it to $V_{t-1}$ to create $V_{t}$. Then, independently, for each vertex $u\in V_{t-1}$ such that $v_t \in S(u,t-1)$, a directed link $(v_{t},u)$ is created with probability $p$. Thus, the probability that a link $(v_t,u)$ is added in time-step $t$ equals $p\,|S(u,t-1)|$.


\subsection{Properties of the model}\label{sec:spa_prop}

In this section, we briefly discuss previous studies on properties and applications of the SPA model. This model is known to produce scale-free networks, which exhibit many of the characteristics of real-life networks \cite{spa1,spa2}. 
Specifically, Theorem~1.1 in~\cite{spa1} proves that the SPA model generates graphs with a power-law in-degree distribution with coefficient $1 + 1/(pA_1)$. On the other hand, the average out-degree is asymptotic to $pA_2/(1-pA_1)$, as shown in Theorem~1.3 in \cite{spa1}. In~\cite{spa3}, it was demonstrated that the SPA model give the best fit, in terms of graph structure, for a series of social networks derived from Facebook. In~\cite{spa4}, some properties of common neighbours were used to explore the underlying geometry of the SPA model and quantify vertex similarity based on the distance in the space.  Usually, the distribution of vertices in $S$ is assumed to be uniform~\cite{spa4}, but~\cite{spa5} also investigated non-uniform distributions, which is clearly a more realistic setting. The SPA model was also used to study a duopoly market on which there is uncertainty of a product quality~\cite{SPAKaminski2017}. Finally, in~\cite{modularity} modularity of this model was investigated, which is a global criterion to define communities and a way to measure the presence of community structure in a network.

\section{Clustering coefficient}\label{sec:clustering}

Clustering coefficient measures how likely two neighbours of a vertex are connected by an edge.
There are several definitions of clustering coefficient proposed in the literature (see, e.g.,~\cite{bollobas2003mathematical}). 
The {\it global clustering coefficient} $C_{glob}(G)$ of a graph $G$ is the ratio of three times the number of triangles to the number of pairs of adjacent edges in $G$. 
In other worlds, if we sample a random pair of adjacent vertices in $G$, then $C_{glob}(G)$ is the probability that these three vertices form a triangle. The global clustering coefficient in the SPA model was previously studied in~\cite{jacob2013spatial,jacob2015spatial} and it was proven that $C_{glob}(G_n)$ converges to a limit, which is positive if and only if the power-law degree distribution has a finite variance.

In this paper, we focus on the {\it local clustering coefficient}, which was not previously analyzed for the SPA model.
Let us first define it for an undirected graph $G = (V, E)$.
Let $N(v)$ be the set of neighbours of a vertex $v$, $|N(v)| = \deg(v)$. For any $B \subseteq V$, let $E(B)$ be the set of edges in the graph induced by the vertex set $B$; that is,
$
E(B) = \{ (u,w) \in E : u, w \in B \}.  
$
Finally, \emph{clustering coefficient} of a vertex $v$ is defined as follows:
$$
c(v) =\frac{|E(N(v))|}{\binom{\deg(v)}{2}}.
$$
Clearly, $0 \le c(v) \le 1$.

Note that the local clustering $c(v)$ is defined individually for each vertex and it can be noisy, especially for the vertices of not too large degrees. Therefore, the following characteristic was extensively studied in the literature for various real-world networks and some random graph models. Let $C(d)$ be the local clustering coefficient averaged over the vertices of degree $d$; that is,
$$
C(d) = \frac{\sum_{v: \deg(v) = d}c(v)}{|\{v: \deg(v) = d\}|}\,.
$$
Further in the paper we will also use the notation $c(v,t)$ and $C(d,t)$ referring to graphs on $t$ vertices.

The local clustering $C(d)$ was extensively studied both theoretically and empirically. 
For example, it was observed in a series of papers that in real-world networks $C(d) \sim d^{-\varphi}$ for some $\varphi > 0$.
In particular, \cite{ravasz2003hierarchical} shows that $C(d)$ can be well approximated by $d^{-1}$ for four large networks, \cite{vazquez2002large} obtains power-law in a real network with parameter 0.75, while \cite{csanyi2004structure} obtains $\varphi = 0.33$.
The local clustering coefficient was also studied in several random graph models of complex networks. For instance, it was shown in \cite{dorogovtsev2002pseudofractal,krot2015local,newman2003properties} that some models have $C(d) \sim d^{-1}$. As we prove in this paper, similar behaviour is also observed in the SPA model.

Recall that the graph $G_t$ constructed according to the SPA model is directed. Therefore, we first analyze the directed version of the local clustering coefficient and then, as a corollary, we obtain the corresponding results for the undirected version. Let us now define the directed clustering coefficient. By $N^-(v,t) \subseteq V_t$ we denote the set of in-neighbours of a vertex $v$ at time $t$. So, the directed clustering coefficient of a vertex $v$ at time $t$ and the average directed clustering for the vertices of incoming degree $d$ are defined as
$$
c^{-}(v,t) = \frac{|E(N^-(v,t))|}{\binom{\deg^-(v,t)}{2}}, \,\,\,\,\,\,\,
C^{-}(d,t) = \frac{\sum_{v: \deg^-(v,t) = d}c^-(v,t)}{|\{v: \deg^-(v,t) = d\}|}\,.
$$
Note that we normalize $c^-(v,t)$ by $\binom{\deg^-(v,t)}{2}$, since in the SPA model edges can be created only from younger vertices to older ones.

\section{Results}\label{sec:results_section}

\subsection{Notation}

Let us start with introducing some notation. As typical in random graph theory, all results in this paper are asymptotic in nature; that is, we aim to investigate properties of $G_n$ for $n$ tending to infinity. We say that an event holds \emph{asymptotically almost surely} (\aas) if it holds with probability tending to one as $n\to\infty$. Also, given a set $S$ we say that \emph{almost all} elements of $S$ have some property $P$ if the number of elements of $S$ that do not have $P$ is $o(|S|)$. Finally, we emphasize that the notations $o(\cdot)$ and $O(\cdot)$ refer to functions of $n$, not necessarily positive, whose growth is bounded. We use the notations $f \ll g$ for $f=o(g)$ and $f \gg g$ for $g=o(f)$. We also write $f(n) \sim g(n)$ if $f(n)/g(n) \to 1$ as $n \to \infty$ (that is, when $f(n) = (1+o(1)) g(n)$). 

First we consider the directed clustering coefficient. It turns out that for the SPA model we are able not only to prove the asymptotics for $C^-(d,n)$, which is the average clustering over all vertices of in-degree $d$, but also analyze the individual clustering coefficients $c^-(v,n)$. However, in order to do this, we need to assume that $\deg^-(v,n)$ is large enough.

From technical point of view, it will be convenient to partition the set of contributing edges, $E(N^-(v,n))$, and independently consider edges to ``old'' and to ``young'' neighbours of $v$. Formally, for a given function $\omega(n)$ that tends to infinity as $n\to \infty$, let $\hat{T}_v$ be the smallest integer $t$ such that $\deg^-(v,t)$ exceeds $\omega \log n$ (or $\hat{T}_v=n$ if $\deg^-(v,n) < \omega \log n$). Vertices in $N^-(v,\hat{T}_v)$ are called \emph{old neighbours of $v$}; $N^-(v,n) \setminus N^-(v,\hat{T}_v)$ are \emph{new neighbours of $v$}. Finally,

$$
E_{old}(N^-(v,n)) = \{ (u,w) \in E_n : u \in N^-(v,n), w \in N^-(v,\hat{T}_v) \},
$$
$$
E_{new}(N^-(v,n)) = E(N^-(v,n)) \setminus E_{old}(N^-(v,n))\,;
$$
and
\begin{equation}\label{eq:old_new}
c^-(v,n) = c_{old}(v,n) + c_{new}(v,n),
\end{equation}
where
\begin{eqnarray*}
c_{old}(v,n) &=& |E_{old}(N^-(v,n))| \Big/ \binom{\deg^-(v,n)}{2}, \\
c_{new}(v,n) &=& |E_{new}(N^-(v,n))| \Big/ \binom{\deg^-(v,n)}{2}.
\end{eqnarray*}

\subsection{Results}\label{sec:results}

Let us start with the following theorem which is extensively used in our reasonings and is interesting and important on its own. Variants of this results were proved in~\cite{spa4,spa5}; here, we present a slightly modified statement from~\cite{spa5}, adjusted to our current needs. 

\begin{theorem}\label{degconc}
Let $\omega=\omega(n)$ be any function tending to infinity together with $n$. The following holds with probability $1-o(n^{-4})$. For any vertex $v$ with
$
\deg^-(v,n)=k=k(n) \geq \omega \log n
$
and for all values of $t$ such that 
$$
n \left(\frac{\omega \log n}{k}\right)^{\frac{1}{p A_1}} =: T_v \le t \le n,
$$
we have
$$
\deg^-( v,t) \sim k \left(\frac{t}{n}\right)^{p A_1}.
$$
\end{theorem}

The expression for $T_v$ is chosen so that at this time vertex $v$ has \aas\ $(1+o(1)) \omega \log n$ neighbours. The implication of this theorem is that once a vertex accumulates $\omega \log n$ neighbours, its behaviour can be predicted with high probability until the end of the process (that is, till time $n$).

Let us note that Theorem~\ref{degconc} immediately implies the following two corollaries. 
\begin{cor}\label{cor:tT}
Let $\omega =\omega (n)$ be any function tending to infinity together with $n$. The following holds with probability $1-o(n^{-4})$. For every vertex $v$, and for every time $T$ so that $\deg^-(v,T)\geq \omega\log n$, for all times $t$, $T\leq t\leq n$,
\[
\deg^-(v,t) \sim \deg^-(v,T) \left( \frac tT \right)^{p A_1}.
\]
\end{cor}

\begin{cor}\label{cor:upper}
Let $\omega =\omega (n)$ be any function tending to infinity together with $n$. The following holds with probability $1-o(n^{-4})$. For any vertex $v_i$ born at time $i \ge 1$, and $i \le t \le n$ we have that
$\deg^-(v_i,t)\leq \omega \log n \left(t/i\right)^{p A_1}.$
\end{cor}

\bigskip

Theorem~\ref{degconc} can be used to show that the contribution to $c^-(v,n)$ coming from edges to new neighbours of $v$ is well concentrated.

\begin{theorem}\label{thm:c_new}
Let $\omega=\omega(n)$ be any function tending to infinity together with $n$. Then, with probability $1-o(n^{-1})$ for any vertex $v$ with 
$$
\deg^-(v,n)=k=k(n) \geq (\omega \log n)^{4 + (4pA_1+2)/(pA_1(1-pA_1))}
$$
we have
$$
c_{new}(v,n) = \Theta(1/k).
$$
\end{theorem}

\bigskip

Unfortunately, if a vertex $v$ lands in a densely populated region of $S$, it might happen that $c_{old}(v,n)$ is much larger than $1/k$. We show the following `negative' result (without trying to aim for the strongest statement) that shows that there is no hope for extending Theorem~\ref{thm:c_new} to $c^-(v,n)$.

\begin{theorem}\label{thm:negative}
Let $C = 5 \log \left( 1/p \right)$ and $\xi = \xi(n) = 1 / (\omega (\log \log n)^2 (\log \log \log n)) = o(1)$ for some $\omega = \omega(n)$ tending to infinity as $n \to \infty$. Suppose that $k = k(n)$ is such that $2 \le k \le n^{\xi}.$
Then, a.a.s., there exists a vertex $v$ such that $\deg^-(v,n) \sim k$ and  
\begin{itemize}
\item [(i)] $c^-(v,n) = 1$, provided that $2 \le k \le \sqrt{\log n / C}$,
\item [(ii)] $c^-(v,n) = \Omega(1) \gg 1/k$, provided that  $\sqrt{\log n / C} \le k \le \log n / \log \log n$,
\item [(iii)] $c^-(v,n) \gg (\log \log n)^2 (\log \log \log n) / k \gg 1/k$, provided that $\log n / \log \log n \le k \le n^{\xi}$.
\end{itemize}
\end{theorem}

On the other hand, Theorem~\ref{thm:c_new} implies immediately the following corollary.

\begin{cor}\label{cor:main}
Let $\omega=\omega(n)$ be any function tending to infinity together with $n$. The following holds with probability $1-o(n^{-1})$. For any vertex $v$ for which 
$$
\deg^-(v,n)=k=k(n) \geq (\omega \log n)^{4 + (4pA_1+2)/(pA_1(1-pA_1))}
$$
it holds that 
\begin{eqnarray*}
c^-(v,n) &\ge& c_{new}(v,n) = \Omega(1/k) \\
c^-(v,n) &=& c_{old}(v,n) + c_{new}(v,n) = O(\omega \log n / k) + O(1/k) = O(\omega \log n / k).
\end{eqnarray*}
\end{cor}

Moreover, despite the above `negative' result, almost all vertices (of large enough degrees) have clustering coefficients of order $1/k$. Here is a precise statement. The conclusions in cases~(i)' and~(ii)' follow immediately from Theorem~\ref{thm:c_new}.

\begin{theorem}\label{thm:average}
Let $\eps, \delta \in (0,1/2)$ be any two constants, and let $k = k(n) \le n^{pA_1 - \eps}$ be any function of $n$. Let $X_k$ be the set of vertices of $G_n$ of in-degree between $(1-\delta)k$ and $(1+\delta)k$. 
Then, a.a.s., the following holds.
\begin{itemize}
\item [(i)] Almost all vertices in $X_k$ have $c_{old}(v,n) = O(1/k)$, provided that $k \gg \log^{C_1} n$, where $C_1 = (1+(2+\eps)pA_1)/(1-pA_1)$.
\item [(i)'] As a result, almost all vertices in $X_k$ have $c^-(v,n) = \Theta(1/k)$, provided that $k \gg \log^{C} n$, where $C = 4 + (4pA_1+2)/(pA_1(1-pA_1))$.
\item [(ii)] The average clustering coefficient $c_{old}(v,n)$ of vertices in $X_k$ is $O(1/k)$; that is,
$$
\frac {1}{|X_k|} \sum_{v \in X_k} c_{old}(v,n) = O(1/k),
$$
provided that $k \gg \log^{C_2} n$, where $C_2 = (1+(2+pA_1+\eps)pA_1)/(1-pA_1)$.
\item [(ii)'] As a result, the average clustering coefficient $c^-(v,n)$ of vertices in $X_k$ is $\Theta(1/k)$; that is,
$$
\frac {1}{|X_k|} \sum_{v \in X_k} c^-(v,n) = \Theta(1/k),
$$
provided that $k \gg \log^{C} n$, where $C = 4 + (4pA_1+2)/(pA_1(1-pA_1))$.
\end{itemize}
\end{theorem}

\bigskip

Finally, let us briefly discuss the undirected case. The following lemma holds.
\begin{lemma}\label{lem:out-degree}
Let $\omega =\omega (n)$ be any function tending to infinity together with $n$. The following holds with probability $1-o(n^{-3})$. For every vertex $v_i$, 
\[
\deg^+(v_i,i) = \deg^+(v_i,n) \le \omega \log n.
\]
\end{lemma}
Note that a weaker bound of $\log^2 n$ was proved in~\cite{spa1}; with Corollary~\ref{cor:upper} in hand, we can get slightly better bound but the argument remains the same. 

According to the above lemma, a.a.s.~the out-degrees of all vertices do not exceed $\omega \log n$. Therefore, even if out-neighbours of a vertex form a complete graph, the contribution from them is at most $\binom{\omega \log n}{2}$, which is much smaller than $k$. Hence, all results discussed in this section also hold for the clustering coefficient $c(k,n)$ defined for the undirected graph $\hat G_n$ obtained from $G_n$ by considering all edges as undirected.

\section{Experiments}

In this section, we illustrate the theoretical, asymptotic, results presented in the previous section by analyzing the local clustering coefficient for graphs of various orders generated according to the SPA model.

\subsection{Algorithm}\label{sec:algorithm}

Let us first discuss the complexity of the straightforward (\textit{naive}) algorithm for generating graphs according to the SPA model. At each step we add one vertex and, for each existing vertex, we check if the new vertex belongs to its sphere of influence. Then we (possibly) add new edges and update the radii for all vertices. The complexity of this procedure is $\Theta(n^2)$.

Let us now propose a more efficient algorithm. First, we describe this algorithm and provide heuristic arguments about its complexity. Then, we compare running times of the new algorithm and the naive one. 

Our algorithm works in several phases, as described further in the text.
For now, let us assume that we already generated a graph on $n$ vertices according to the SPA model and we want to add one additional vertex.
It is known that
$$
\E \Big( \deg^- (v_i, t) \Big) \sim \frac {A_2}{A_1} \left( \frac {t}{i} \right)^{p A_1},
$$ 
provided that $i \gg 1$ (see, for example,~\cite{spa2}). We call a vertex \emph{heavy} if its degree is at least $D$ for some $D$; otherwise, it is \emph{light}. All heavy vertices are kept in a separate list $H$. Fix 
\begin{equation}\label{eq:D} 
D = \frac {A_2}{A_1} \left( \frac {n}{T} \right)^{pA_1},
\end{equation}
so $H$ has expected size around $T$. The choice of an optimal value of $T$ will be discussed further in this section.

Let us divide $S = [0,1]^2$ into $k$ squares where $k$ is some perfect square; that is, each square will have side length $1/\sqrt{k}$. (We choose the dimension $m=2$ for our simulations. However, the ideas can easily be applied for an arbitrary $m$.) All light vertices are kept in $k$ disjoint lists; let $L(i)$ be a list containing all light vertices from square $i$. The expected number of vertices in each list is $(n-T)/k$.

We want the following property to be satisfied:
\begin{equation}\label{eq:D}
\sqrt{ \frac {A_1 D + A_2}{\pi n} } \le \frac {1}{\sqrt{k}}.
\end{equation}
Indeed, if this is the case, then no light vertex $v_i$ has the area of influence that touches squares other than the square containing $v_i$ and the 8 adjacent squares. Moreover, the same property will hold for all $t>n$ as areas of influence of light vertices decrease with time. Hence, since we aim for an integer $\sqrt{k}$ to be as large as possible:
\begin{equation}\label{eq:k}
k = \left\lfloor \sqrt{\frac {\pi n}{A_1 D + A_2}} \right\rfloor^2 \Rightarrow k \approx \frac {\pi n} {A_2 \left(1+(n/T)^{pA_1}\right)}.
\end{equation}

The most expensive computational work for the algorithm is the number of comparisons needed in order to add a vertex $v_{n+1}$ to a graph, which is of order
\begin{equation}\label{eq:f}
f(T) := T + 9 \ \frac{n-T}{k} = 
T + \frac{9A_2}{\pi}(1-T/n)\left(1+(n/T)^{pA_1}\right)\,.
\end{equation}
Hence, the function $f(T)$ is minimized for
\begin{equation*}
T = \frac{9npA_1A_2(n/T)^{pA_1}}{\pi n - 9A_2 - 9A_2(1-pA_1)(n/T)^{pA_1}}\,.
\end{equation*}
For large $n$ the second and the third terms in the denominator are negligible, as $pA_1<1$; moreover, if $pA_1$ is close to $1$ we will soon show that $T=\Theta(n^{1/2-\epsilon})$ for some small $\epsilon >0$, so the approximation converges fast. Thus, we may approximate $T$ by:
\begin{equation}
T \approx n^{1-1/(pA_1+1)}\left(\frac{9pA_1A_2}{\pi}\right)^{1/(pA_1+1)}.\label{eq:optT}
\end{equation}
Using this $T$ we can calculate the recommended value of $D$, see~(\ref{eq:D}), and the density of the $\sqrt{k} \times \sqrt{k}$ grid, see~(\ref{eq:k}).

Below are some practical implementation details:
\begin{itemize}
\item It is computationally expensive to recalculate $H$ and $L$ division each time a new vertex is added.  By empirical testing, we have found that the recalculation should be done approximately after adding $t/4$ vertices, where $t$ is the number of vertices in already constructed graph. As a result, the number of phases is $O(\log n)$, as each time the number of vertices increases by approximately 25\%.
\item As we work in phases, at each step we have to check if some light vertex becomes heavy, and move it to the appropriate list, if needed. However, this operation is not expensive computationally.
\item After several phases, for actually constructed graphs the optimal parameters $k$, $T$ and $D$ might deviate from the theoretical values derived above. Therefore, in the implementation we choose the optimal parameters conditional on the actual input graph structure. Namely, for each candidate value $k$ we can calculate the corresponding $D$ using~\eqref{eq:D} and then calculate $T$ from the data (this is the actual number of heavy vertices given $D$). We choose $k$ to optimize the number of comparisons needed to add one vertex to the actual graph, the approximation for this value is given in~\eqref{eq:f}. After that we dynamically construct $H$ and $L$ lists. 
\end{itemize}
Let us now discuss the complexity of the obtained algorithm. Equation~\eqref{eq:f} shows that $T$ is expected to be of order $n^{pA_1/(pA_1+1)}$. So, we may derive from~\eqref{eq:k} that $k$ is of order $n^{1-pA_1 + (pA_1)^2/(pA_1+1)} = n^{1/(pA_1+1)}$. From~\eqref{eq:f} we obtain that $f(T)$ grows as $n^{pA_1/(pA_1+1)}$. So, the expected complexity of the whole process is $\Theta\left(n^{2 - 1/(pA_1+1)}\right) \ll n^2$.

Figure~\ref{fig:runtime} presents an empirical comparison of the running time for new and naive algorithms. We also present this figure in log-log scale. The computations were performed using Julia 0.6.2 language \cite{julia} and LightGraphs \cite{lightgraphs} package on a single thread of Intel i5-5200U @ 2.20GHz processor.

\begin{figure}[h]
\begin{center}
\includegraphics[height=4.82cm]{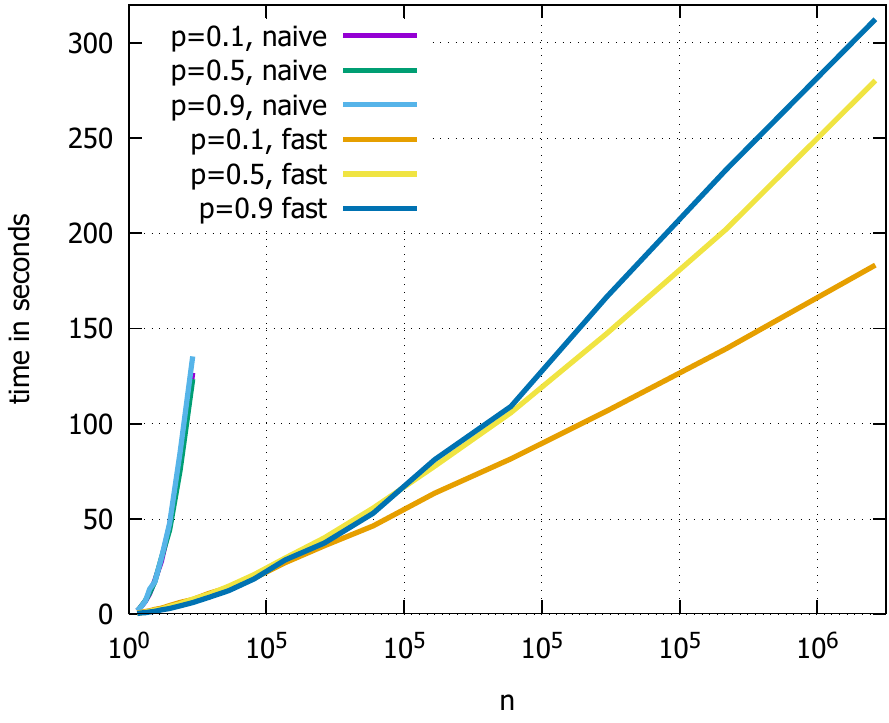}
\includegraphics[height=4.82cm]{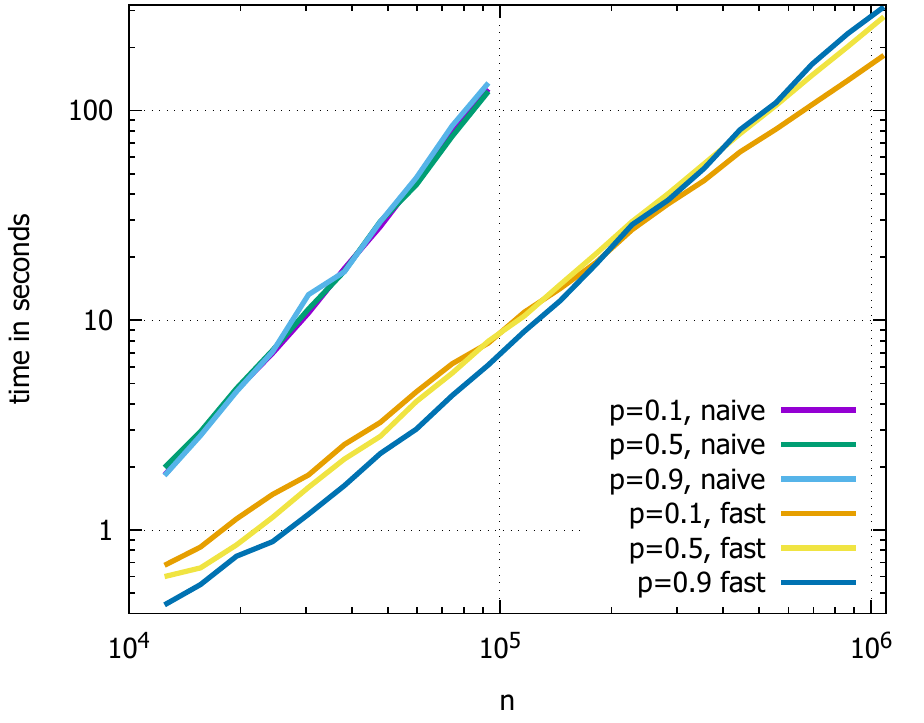}
\vspace{-20pt}
\caption{Running time of the proposed and the naive algorithms.}\label{fig:runtime}
\end{center}
\vspace{-14pt}
\end{figure}

Finally, let us mention that further improvements of the algorithm are possible. For example, one can keep more than two lists $H$ and $L$. For example, $L_s(i)$ could contain vertices of degree between $2^{s-1}$ and $2^{s}$ that landed in region $i$, so the total number of lists is $O(\log n)$. Then, the running time of the algorithm would be $O(n \log n)$. Indeed, during a phase that started at time $t$, $L_s$ has expected size $O(t\, 2^{-s/pA_1})$; since vertices from $L_s(i)$ are gathered from the square of area, say, $2^s/t$, the expected size of this list is $O(2^{s-s/(pA_1)}) = O(1)$. Hence, after adding one vertex, $O(\log n)$ lists are investigated and we expect only a constant number of comparisons done on each list. Of course, there is always a trade-off between the running time of an algorithm and how complicated it is to implement it. For our purpose, we decided to go for a simpler algorithm with only two lists.


\subsection{Empirical analysis of the local clustering coefficient}

In this section, we compare asymptotic theoretical results obtained in Section~\ref{sec:results_section} with empirical results obtained for graphs with finite $n$. All graphs are generated according to the algorithm described in Section~\ref{sec:algorithm}.

It is proven in Theorem~\ref{thm:average} that $\frac{1}{X_d} \sum_{v \in X_d} c^-(v,n) = \Theta (1/d)$ for $d \gg \log^C n$, where $C = 4 + (4pA_1 + 2)/(pA_1(1-pA_1))$.
In order to illustrate this result, we generated 10 graphs for each $p \in \{0.1,0.2,\ldots,0.9\}$, $A_1 = 1$, $A_2 = 10(1-p)/p$ ($A_2$ is chosen to fix the expected asymptotic degree equal 10) and computed the average value of $C^-(d,n)$ for $n=10^6$, see Figure~\ref{fig:average_clustering} (left). Similarly, Figure~\ref{fig:average_clustering} (right) presents the same measurements for the undirected average local clustering $C(d,n)$.
Note that in both cases figures agree with our theoretical results: both $C^-(d,n)$ and $C(d,n)$ decrease as $c/d$ with some $c$ for large enough $d$ (we added a function $10/d$ for comparison).
Note that for small $p$ the maximum degree is small, therefore the sizes of the generated graphs are not large enough to observe a straight line in log-log scale. 

\begin{figure}
\begin{center}
\includegraphics[width=0.49\textwidth]{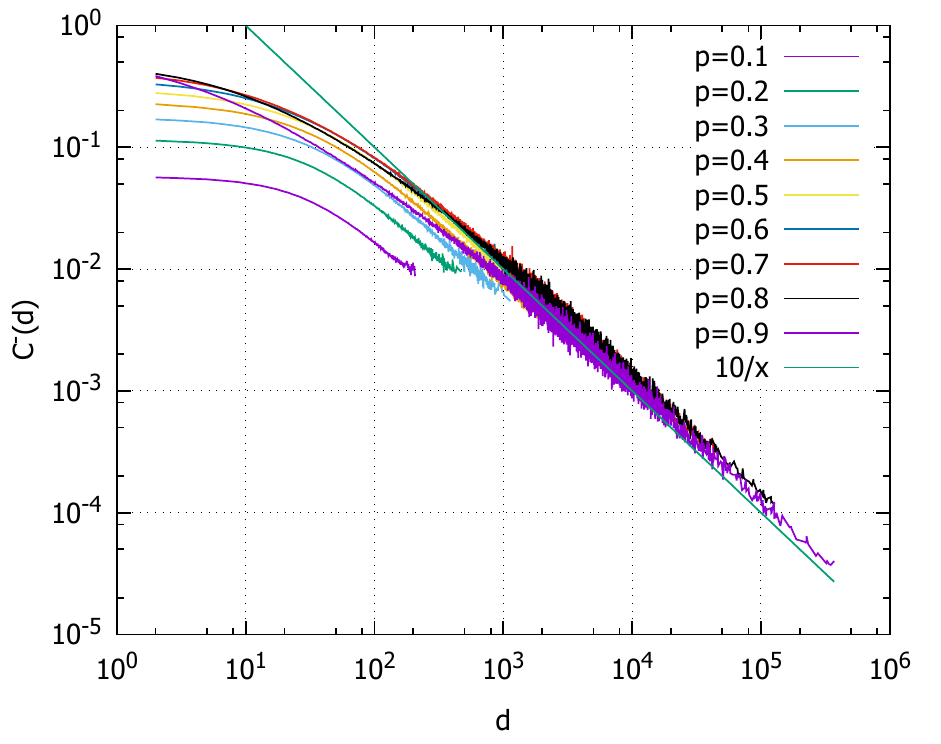}
\includegraphics[width=0.49\textwidth]{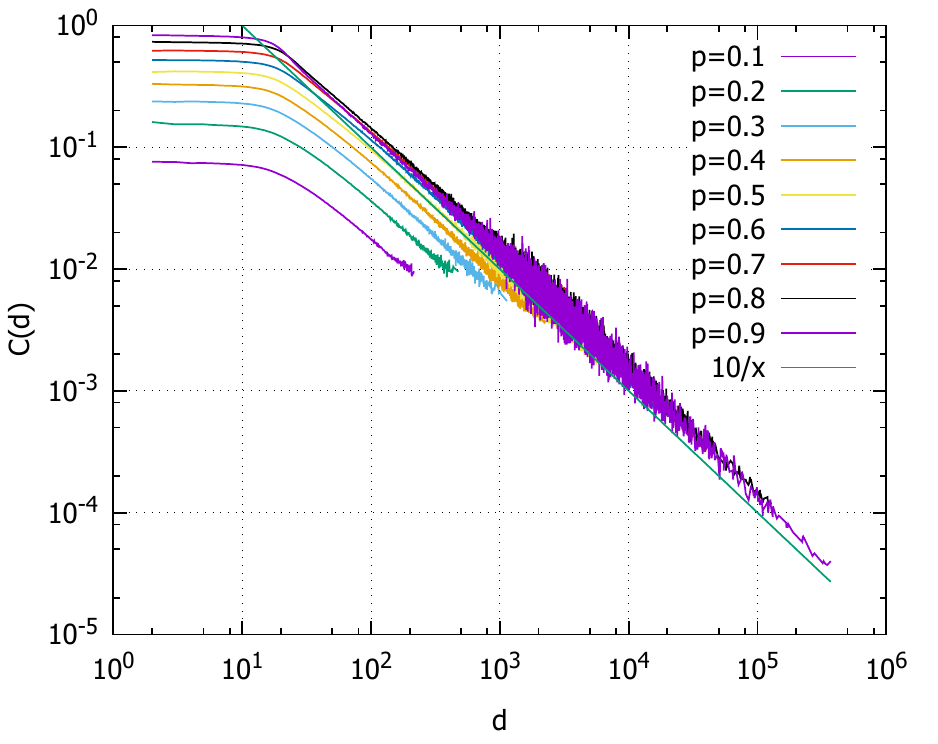}
\vspace{-10pt}
\caption{Average local clustering coefficient for directed (left) and undirected (right) graphs.}
\label{fig:average_clustering}
\end{center}
\vspace{-20pt}
\end{figure}

Note that for all $p\in (0,1)$ we have $C = 4 + \frac{4p+2}{p(1-p)} > 18$, so, our theoretical results are expected to hold for $d \gg \log^C n > 10^{20}$ which is irrelevant as the order of the graph is only $10^{6}$. However, we observe the desired behaviour for much smaller values of $d$; that is, in some sense, our bound is too pessimistic.
 

Also, note that the statement $C^-(d,n) = \Theta(1/d)$ is stronger that the statement of Theorem~\ref{thm:average}, since in the theorem we averaged $c^{-}(v,n)$ over the set $X_d$ of vertices of in-degree between $(1-\delta)d$ and $(1+\delta)d$. In order to illustrate the difference, on Figure~\ref{fig:smooth_clustering} we present the smoothed curves for the directed (left) and undirected (right) local clustering coefficients averaged over $X_d$ for $\delta = 0.1$. Note that this smoothing substantially reduce the noise observed on Figure~\ref{fig:average_clustering}.

\begin{figure}
\begin{center}
\includegraphics[width=0.49\textwidth]{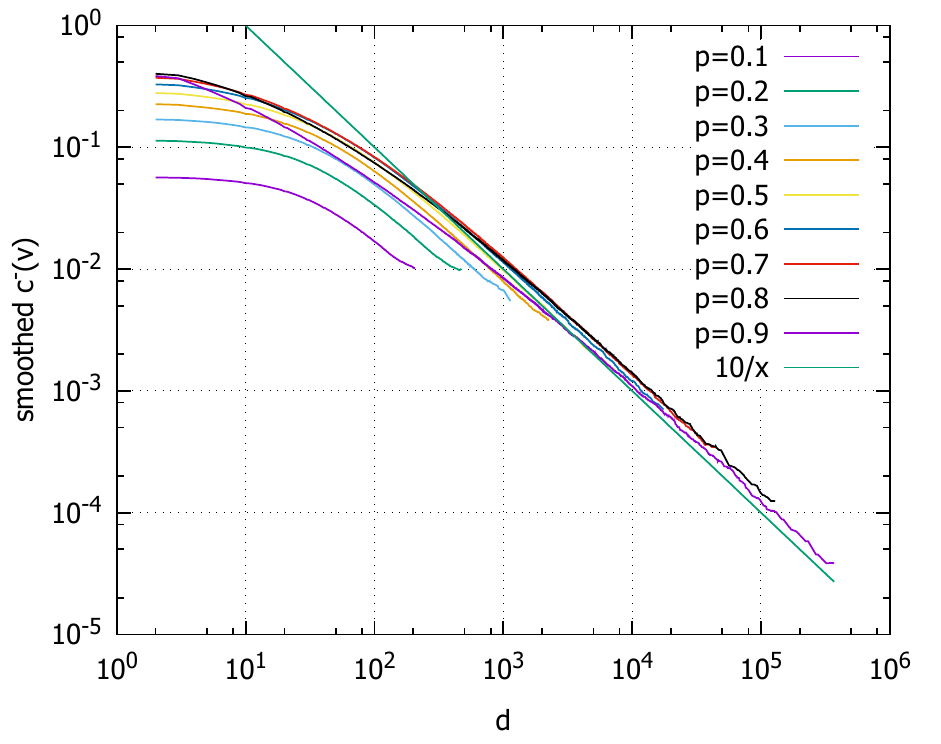}
\includegraphics[width=0.49\textwidth]{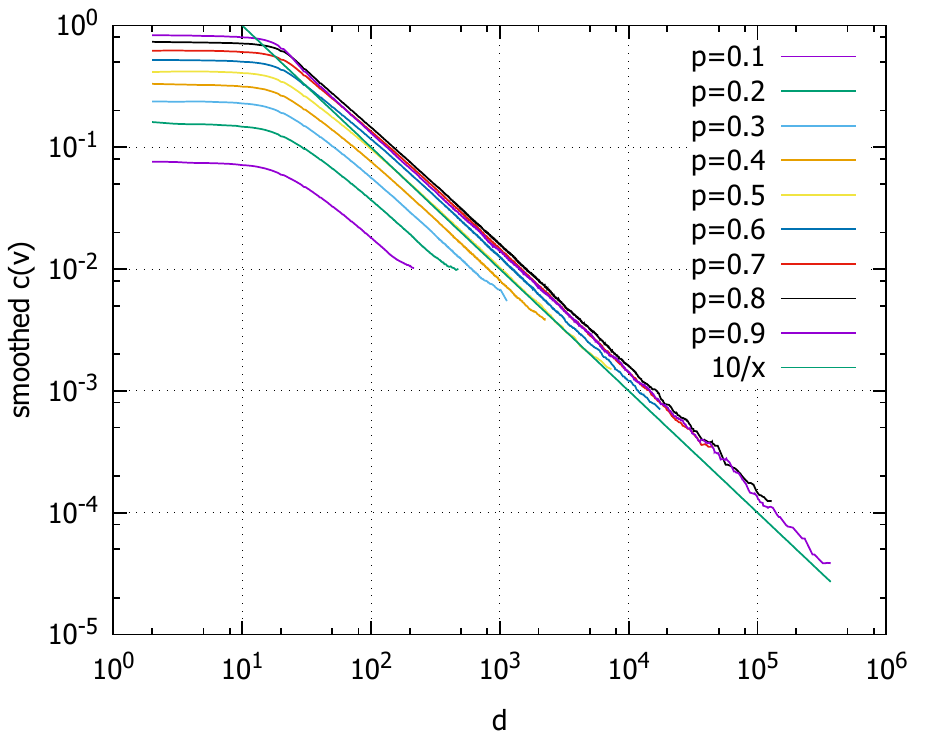}
\vspace{-10pt}
\caption{Local clustering coefficient for directed (left) and undirected (right) graphs averaged over $X_d$.}
\label{fig:smooth_clustering}
\end{center}
\vspace{-20pt}
\end{figure}

Next, let us illustrate the fact that the number of edges between ``new'' neighbours of a vertex is more predictable than the number of edges going from some neighbours to ``old'' ones. We extensively used this difference in Section~\ref{sec:results}, where we analyzed new and old edges separately. 
In our experiments, we split $c^-(v,n)$ into ``old'' and ``new'' parts as in~\eqref{eq:old_new}, but now we take $\hat{T}_v$ be the smallest integer $t$ such that $\deg^-(v,t)$ exceeds $\deg^-(v,n)/2$. As a result, we compute the average local clustering coefficients $C_{old}^-(d)$ and $C_{new}^-(d)$. Figure~\ref{fig:old_new} shows that $C_{new}^-(d)$ can almost perfectly be fitted by $c/d$ with some $c$, while most of the noise comes from  $C_{old}^-(d)$. 

\begin{figure}
\begin{center}
\includegraphics[width=0.49\textwidth]{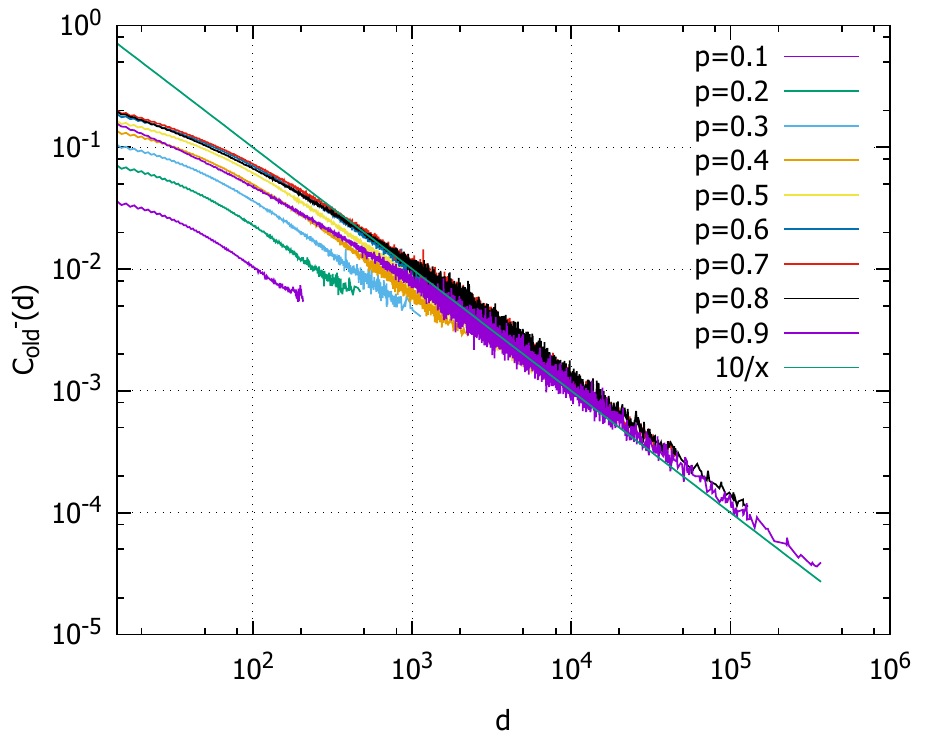}
\includegraphics[width=0.49\textwidth]{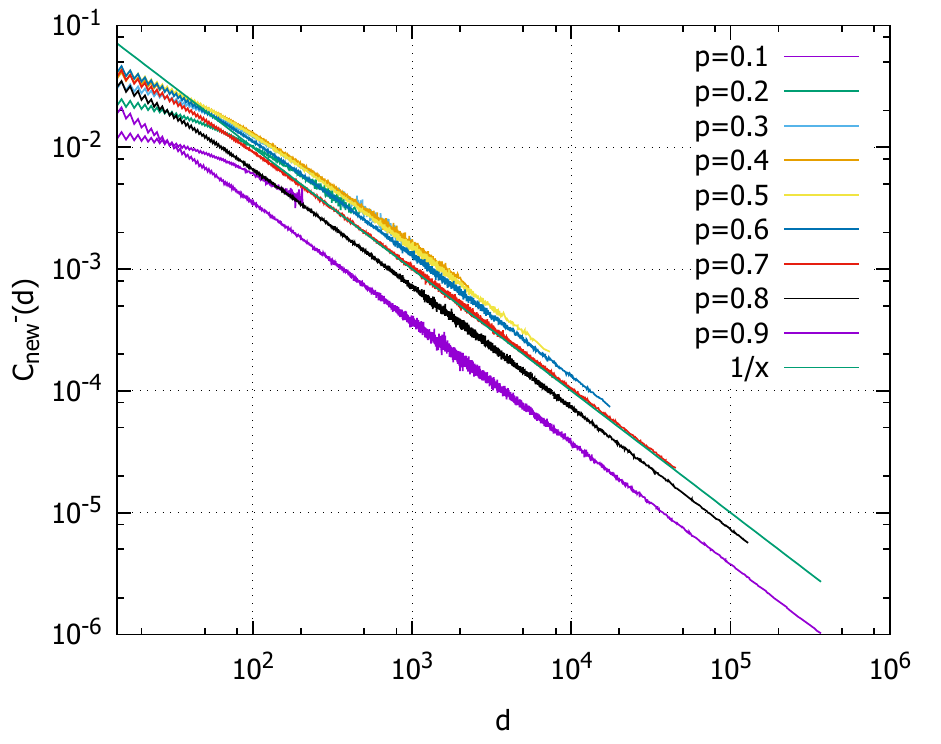}
\caption{Comparison of ``new'' and ``old'' parts of the average local clustering coefficient.}
\label{fig:old_new}
\end{center}
\end{figure}

Finally, Figure~\ref{fig:individual} shows the distribution of the individual local clustering coefficients for one graph generated with $p=0.7$. Theorem~\ref{thm:negative} states that a.a.s. there exist a vertex $v$ of degree $d$ with $c^-(v,n) \gg 1/d$. Also, according to this theorem, the situation is much worse for smaller values of $d$. 
Indeed, one can see on Figure~\ref{fig:individual} that for small $d$ the scatter of points is much larger. On the other hand, in Theorem~\ref{thm:average} we present bounds for $c^-(v,n)$ for almost all vertices, provided that $d$ is large enough. One can see it on the figure too and, similarly to previously discussed figures, we observe the expected behaviour even for relatively small $n$ despite the bound $\log^C n$ that is bigger than $n$ in our case.

\begin{figure}
\begin{center}
\includegraphics[width=0.6\textwidth]{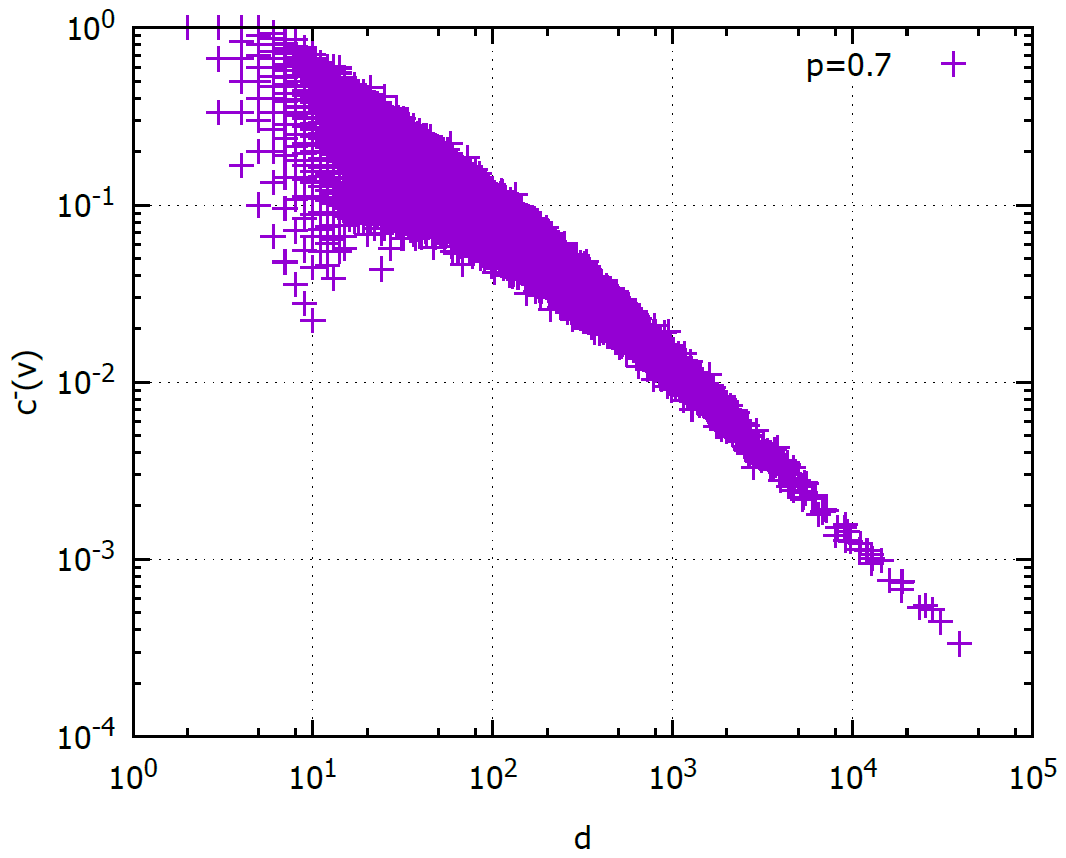}
\caption{The distribution of individual local clustering coefficients.}
\label{fig:individual}
\end{center}
\end{figure}

\section*{Acknowledgements}

This work is supported by Russian President grant MK-527.2017.1 and NSERC.

\bibliographystyle{splncs03}
\bibliography{SPA}

\end{document}